\documentclass[a4paper,10pt]{article}
\usepackage[utf8]{inputenc}
\usepackage{amsmath, amsthm, amssymb}
\usepackage{pstricks,pst-node,pst-plot}
\usepackage{enumerate}
\usepackage{booktabs}
\usepackage{todonotes}
\reversemarginpar

\title{A new characterization of $P_k$-free graphs}
\author{Eglantine Camby\\
Universit\'e Libre de Bruxelles\\
D\'epartement de Math\'ematique\\
Boulevard du Triomphe, 1050 Brussels, Belgium\\
\texttt{ecamby@ulb.ac.be}\\~\\
Oliver Schaudt\\
University of Cologne\\
Institute for Computer Science\\
Weyertal 80, 50321 Cologne, Germany\\
\texttt{schaudto@uni-koeln.de}}

\begin{document}

\newtheorem{theorem}{Theorem}
\newtheorem{lemma}{Lemma}
\newtheorem{observation}{Observation}
\newtheorem{corollary}{Corollary}
\newtheorem{claim}{Claim}
\newtheorem{conjecture}{Conjecture}

\maketitle

\begin{abstract}
\noindent 
The class of graphs that do not contain an induced path on $k$ vertices, $P_k$-free graphs, plays a prominent role in algorithmic graph theory.
This motivates the search for special structural properties of $P_k$-free graphs, including alternative characterizations.

Let $G$ be a connected $P_k$-free graph, $k \ge 4$.
We show that $G$ admits a connected dominating set whose induced subgraph is either $P_{k-2}$-free, or isomorphic to $P_{k-2}$.
Surprisingly, it turns out that every minimum connected dominating set of $G$ has this property.

This yields a new characterization for $P_k$-free graphs:
a graph $G$ is $P_k$-free if and only if each connected induced subgraph of $G$ has a connected dominating set whose induced subgraph is either $P_{k-2}$-free, or isomorphic to $C_k$.
This improves and generalizes several previous results; the particular case of $k=7$ solves a problem posed by van 't Hof and Paulusma [A new characterization of $P_6$-free graphs, COCOON 2008].

In the second part of the paper, we present an efficient algorithm that, given a connected graph $G$ on $n$ vertices and $m$ edges, computes a connected dominating set $X$ of $G$ with the following property:
for the minimum $k$ such that $G$ is $P_k$-free, the subgraph induced by $X$ is $P_{k-2}$-free or isomorphic to $P_{k-2}$.

As an application our results, we prove that \textsc{Hypergraph 2-Colorability}, an NP-complete problem in general, can be solved in polynomial time for hypergraphs whose vertex-hyperedge incidence graph is $P_7$-free.

\noindent \textbf{keywords:} $P_k$-free graph, connected domination, computational complexity.

\end{abstract}

\section{Introduction}

A \emph{dominating set} of a graph $G$ is a vertex subset $X$ such that every vertex not in $X$ has a neighbor in $X$.
Dominating sets have been intensively studied in the literature.
The main interest in dominating sets is due to their relevance on both theoretical and practical side.
Moreover, there are interesting variants of domination and many of them are well-studied.
%An introduction into the mathematical aspects of this topic is given by Haynes, Hedetniemi and Slater~\cite{haynes}.

A \textit{connected dominating set} of a graph $G$ is a dominating set $X$ whose induced subgraph, henceforth denoted $G[X]$, is connected.
%The minimum size of such a set of a connected graph $G$, the \textit{connected domination number} of $G$, is denoted by $\gamma_c(G)$.
As usual, a connected dominating set such that every proper subset is not a connected dominating set is called a \emph{minimal connected dominating set}.
A connected dominating set of minimum size is called a \emph{minimum connected dominating set}.
%Among the applications of connected dominating sets is the routing of messages in mobile ad-hoc networks. 
%Blum, Ding, Thaeler and Cheng \cite{MANETS} explain the usefulness of connected dominating sets in this context.

We use the following standard notation.
Let $P_k$ be the induced path on $k$ vertices and let $C_k$ be the induced cycle on $k$ vertices.
If $G$ and $H$ are two graphs, we say that $G$ is \emph{$H$-free} if $H$ does not appear as an induced subgraph of $G$.
Furthermore, if $G$ is $H_1$-free and $H_2$-free for some graphs $H_1$ and $H_2$, we say that $G$ is \emph{$(H_1,H_2)$-free}. If two graphs $G$ and $H$ are isomorphic, we write $G \cong H$.

The class of $P_k$-free graphs has received a fair amount of attention in the theory of graph algorithms.
Given an NP-hard optimization problem, it is often fruitful to study its complexity when the instances are restricted to $P_k$-free graphs.

Let us mention two recent results in this direction: the polynomial time algorithm to compute a stable set of maximum weight, given by Lokshtanov \emph{et al.}~\cite{LVV14}, and the result of Hoang~\emph{et al.}~\cite{HKLSS65} showing that $k$-\textsc{Colorability} is efficiently solvable on $P_5$-free graphs.
The proof of the latter result relies on the fact that a connected $P_5$-free graph has a dominating clique or a dominating $P_3$.

\begin{theorem}[B\'acso and Tuza~\cite{BT90b}]\label{thm:P5}
Let $G$ be a connected $P_5$-free graph.
Then $G$ has a dominating clique or a dominating induced $P_3$.
\end{theorem}

%Note that a dominating clique is simply a connected dominating set whose induced subgraph is $P_3$-free.
%So, Theorem~\ref{thm:P5} says that a connected $P_5$-free graph has a connected dominating set whose induced subgraph is either $P_3$-free, or isomorphic to $P_3$.
An immediate implication of this result is the following.

\begin{theorem}[B\'acso and Tuza~\cite{BT90b}, Cozzens and Kelleher~\cite{CK90}]\label{thm:P5chara}
Let $G$ be a graph. The following assertions are equivalent.
\begin{enumerate}[(i)]
	\item $G$ is $P_5$-free.
	\item\label{cond:P3-free} Every induced subgraph $H$ of $G$ admits a connected dominating set $X$ such that $H[X]$ is a clique or $H[X] \cong C_5$.  
\end{enumerate}
\end{theorem}

%Liu \emph{et al.}~\cite{LPZ07} obtained a characterization for the class of $P_6$-free graphs in the flavour of Theorem~\ref{thm:P5chara}.
%For their characterization, they introduce \emph{triangle extended complete bipartite} graphs: these are the graphs obtained from a complete bipartite graph
%$G$ by adding some extra vertices $v_1,v_2 \ldots, v_k$ and edges $v_iu,v_i v$ for $1 \le i \le k$ to a fixed edge $uv$ of $G$.
%They show that a graph $P_6$-free if and only if each connected induced subgraph contains a dominating induced $C_6$ or connected dominating set whose induced subgraph is spanned by a  triangle extended complete bipartite graph.

Later, van 't Hof and Paulusma~\cite{HP10} obtained a characterization for the class of $P_6$-free graphs in the flavour of Theorem~\ref{thm:P5chara}.
An earlier, slightly weaker result was given by Liu \emph{et al.}~\cite{LPZ07}, and the particular case of triangle free graphs was discussed before by Liu and Zhou~\cite{LZ94}.

\begin{theorem}[van 't Hof and Paulusma~\cite{HP10}]\label{thm:P6}
Let $G$ be a graph. The following assertions are equivalent.
\begin{enumerate}[(i)]
	\item $G$ is $P_6$-free.
	\item\label{cond:completespanner} Every induced subgraph $H$ of $G$ admits a connected dominating set $X$ such that $H[X]$ has a complete bipartite spanning subgraph or $H[X] \cong C_6$.  
\end{enumerate}
\end{theorem}

Complementing Theorem~\ref{thm:P6}, van 't Hof and Paulusma give a polynomial time algorithm that, given a connected $P_6$-free graph, computes a connected dominating set $X$ such that $G[X]$ has a complete bipartite spanning subgraph or $G[X] \cong C_6$.
\medskip

In view of Theorems~\ref{thm:P5chara} and~\ref{thm:P6}, two questions arise.
The first one is whether condition~(\ref{cond:completespanner}) of Theorem~\ref{thm:P6} can be tightened, such that $H[X]$ is a $P_4$-free graph or $G[X] \cong C_6$.
Note that if $H[X]$ is $P_4$-free, it is a connected cograph, and in particular has a complete bipartite spanning subgraph.
This condition is the direct analogue of condition~(\ref{cond:P3-free}) of Theorem~\ref{thm:P5chara} for $P_6$-free graphs.
The advantage of the strenghtened version is of course that the structure of cographs is well understood and more restricted compared to the class of graphs having a spanning complete bipartite graph.

The second question is whether similar characterizations can be given for the class of $P_k$-free graphs, for $k > 6$.
In their paper, van 't Hof and Paulusma~\cite{HP10} explicitly ask for such a characterization in the case of $k=7$.

\subsection{Our contribution}

In this paper, we give an affirmative answer to these two questions.
We show that every connected $P_k$-free graph, $k \ge 4$, admits a connected dominating set whose induced subgraph is either $P_{k-2}$-free, or isomorphic to $P_{k-2}$.
Surprisingly, it turns out that every minimum connected dominating set has this property.

\begin{theorem}\label{thm:sufficiency}
Let $G$ be a connected $P_k$-free graph, $k \ge 4$, and let $X$ be any minimum connected dominating set of $G$. Then $G[X]$ is $P_{k-2}$-free, or $G[X] \cong P_{k-2}$. 
\end{theorem}

From this result we derive the following characterization of $P_k$-free graphs.

\begin{theorem}\label{thm:PkCharacterization}
Let $G$ be a graph and $k \ge 4$. The following assertions are equivalent.
\begin{enumerate}[(i)]
	\item\label{ass:Pk-free} $G$ is $P_k$-free.
	\item\label{ass:k-short-CDS} Every connected induced subgraph $H$ of $G$ admits a connected dominating set $X$ such that $H[X]$ is $P_{k-2}$-free or $H[X] \cong C_k$. 
\end{enumerate}
\end{theorem}

We now come to the algorithmic dimension of the problem.
The proof of Theorem~\ref{thm:sufficiency} is constructive in the sense that it yields an algorithm to compute, given a $P_k$-free graph, a connected dominating set whose induced subgraph is either $P_{k-2}$-free, or isomorphic to $P_{k-2}$.
However, recall that the computation of a longest induced path in a graph is an NP-hard problem, as shown in Garey and Johnson~\cite[p.~196]{GJ79}.
In other words, there is little hope of computing in polynomial time the minimum $k$ for which the input graph is $P_k$-free.  
To overcome this obstacle, our algorithm can only make implicite use of the absent induced $P_k$, which is the main difficulty here.

\begin{theorem}\label{thm:algo}
Given a connected graph $G$ on $n$ vertices and $m$ edges, one can compute in time $\mathcal O (n^5(n+m))$ a connected dominating set $X$ with the following property:
for the minimum $k \ge 3$ such that $G$ is $P_k$-free, $G[X]$ is $P_{k-2}$-free or $G[X] \cong P_{k-2}$.
\end{theorem}

Our last result is an application of the previous theorems.
A \emph{2-coloring} of a hypergraph assigns to each vertex one of two colors, such that each hyperedge contains vertices of both colors.
The problem \textsc{Hypergraph 2-Colorability} is to decide whether a given hypergraph admits a 2-coloring. Garey and Johnson~\cite[p.~221]{GJ79} explain that it is NP-complete in general.
One successful approach to deal with this hardness is to put restrictions on the bipartite vertex-hyperedge incidence graph\footnote{Recall that for a hypergraph $H = (V,E)$ we define the bipartite vertex-hyperedge incidence graph as the bipartite graph on the set of vertices $V \cup E$ with the edges $vY$ such that $v \in V$, $Y \in E$ and $v \in Y$. In the following, we just say the \emph{incidence graph}.} of the input hypergraph.

As an application of Theorem~\ref{thm:P6}, van 't Hof and Paulusma~\cite{HP10} show that \textsc{Hypergraph 2-Colorability} is solvable in polynomial time for hypergraphs with $P_6$-free incidence graph.
Using our results, we settle the case of hypergraphs with $P_7$-free incidence graph.

\begin{theorem}\label{thm:2-coloring}
\textsc{Hypergraph 2-Colorability} can be solved in polynomial time for hypergraphs with $P_7$-free incidence graph.
If it exists, a 2-coloring can be computed in polynomial time.
\end{theorem}

The proof of our results we give in the subsequent sections.
We close the paper with a short discussion of our contribution.

\section{Proofs}

\subsection{Proof of Theorems~\ref{thm:sufficiency} and~\ref{thm:PkCharacterization}}

We need the following lemma from an earlier paper of ours~\cite{CS13}.

\begin{lemma}[Camby and Schaudt~\cite{CS13}]\label{lem:PkCk}
Let $G$ be a connected graph that is $(P_k,C_k)$-free, for some $k \ge 4$, and let $X$ be a minimal connected dominating set of $G$.
Then $G[X]$ is $P_{k-2}$-free.
\end{lemma}

When applied to $P_k$-free graphs, which are in particular $(P_{k+1},C_{k+1})$-free, the above lemma implies that any minimal connected dominating set induces a $P_{k-1}$-free graph, for $k \ge 3$.
We next prove a simple but useful lemma, which plays a key role also in the proof of Theorem~\ref{thm:algo}.
%Basically it shows how a that a minimal connected dominating set containing an induced $P_{k-2}$ can be deformed into a path by the successive iteration.
Let $X$ be a connected dominating set of a graph $G$, and $x \in X$.
Assuming that $X$ is a minimal connected dominating set and $|X|\ge 2$, $x$ is a cut-vertex of $G[X]$ or $x$ has a \emph{private neighbor}:
a vertex $y \in V(G) \setminus X$ with $N_G(y) \cap X = \{x\}$.

\begin{lemma}\label{lem:substitution}
Let $G$ be a $P_k$-free graph, for some $k \ge 4$, and let $X$ be a minimal connected dominating set of $G$.
Assume that there is an induced $P_{k-2}$ in $G[X]$, say on the vertices $x_1,x_2,\ldots,x_{k-2}$.
Then any private neighbor $y$ of $x_1$ is such that $(X \cup \{y\}) \setminus \{x_{k-2}\}$ is a connected dominating set of $G$.
\end{lemma}
\begin{proof}
Note that $G$ is in particular $(P_{k+1},C_{k+1})$-free and thus, by Lemma~\ref{lem:PkCk}, $G[X]$ is $P_{k-1}$-free.

Let $X' := \{x_1,x_2,\ldots,x_{k-2}\}$.
Moreover, let $y$ be any private neighbor of $x_1$, and let $Y := (X \cup \{y\}) \setminus \{x_{k-2}\}$.
We have to prove that $Y$ is a connected dominating set of $G$.

Suppose for a contradiction that $G[Y]$ is not connected.
Hence, $x_{k-2}$ is a cut-vertex of $G[X]$.
In particular, there is some vertex $y' \in X$ such that $N_G(y') \cap X' = \{x_{k-2}\}$.
But then $G[{X'\cup\{y'\}}] \cong P_{k-1}$, a contradiction.

It remains to show that $Y$ is a dominating set.
Suppose the contrary, that is, there is some vertex $x'$ with $N_G[x'] \cap Y = \emptyset$.
As $X$ is a dominating set, $N_G[x'] \cap X = \{x_{k-2}\}$.
Because $x_{k-2}$ is adjacent to $Y$ and $x'$ is not adjacent to $Y$, $x' \neq x_{k-2}$.
But this means that $G[X' \cup \{y,x'\}] \cong P_k$, a contradiction.
\end{proof}

Now we can state the proof of Theorem~\ref{thm:sufficiency}.

\begin{proof}[Proof of Theorem~\ref{thm:sufficiency}.]
Let $X$ be a minimum connected dominating set of $G$.
As $G$ is in particular $(P_{k+1},C_{k+1})$-free, $G[X]$ is $P_{k-1}$-free, by Lemma~\ref{lem:PkCk}.
We have to show that $G[X]$ is $P_{k-2}$-free or isomorphic to $P_{k-2}$.
%From this it follows that if $G[X]$ is not $P_{k_2}$-free

To see this, assume there is an induced $P_{k-2}$ in $G[X]$, say on the vertices $x_1,x_2,\ldots,x_{k-2}$.
Let $X':=\{x_1,x_2,\ldots,x_{k-2}\}$.
Note that $x_1$ is not a cut-vertex of $G[X]$: otherwise there is some vertex $y' \in X$ such that $N_G(y') \cap X' = \{x_1\}$, and hence $G[{X'\cup\{y'\}}] \cong P_{k-1}$.
This is a contradiction.
Thus, $x_1$ is not a cut-vertex of $G[X]$ and therefore has a private neighbor w.r.t.~$X$, say $y_1$.
By Lemma~\ref{lem:substitution}, $Y_1:=(X \cup \{y_1\}) \setminus \{x_{k-2}\}$ is a connected dominating set of $G$.
As $X$ is a minimum connected dominating set, $Y_1$ is a minimum connected dominating set, too.
Moreover, $y_1$ has no neighbor in $X \setminus \{x_1\}$, in particular in $X \setminus X'$.

By reapplying the argumentation to $Y_1$ and the induced $P_{k-2}$ on $y_1,x_1,x_2,\ldots,x_{k-3}$,
We obtain a vertex $y_2 \in V(G) \setminus Y_1$ such that $Y_2:=(Y_1 \cup \{y_2\}) \setminus \{x_{k-3}\}$ is a minimum connected dominating set of $G$ and $G[Y_2]$ contains an induced $P_{k-2}$ on the vertices $y_2,y_1,x_1,x_2,\ldots,x_{k-4}$.
Moreover, $y_2$ has no neighbor in $Y_1 \setminus \{y_1\}$, in particular in $X \setminus X'$.

Iteratively, we end up with a minimum connected dominating set $Y_{k-2}$, which is exactly $(X \setminus X') \cup \{y_1, \dots, y_{k-2}\}$. 
Since, for $i = 1,2, \ldots, k-2$, $y_i$ is not adjacent to $X \setminus X'$ and $G[Y_{k-2}]$ is connected, $X \setminus X'$ must be empty, hence $X=X'$.
Thus, $G[X] = G[X'] \cong P_{k-2}$.
This completes the proof.
\end{proof}

\begin{proof}[Proof of Theorem~\ref{thm:PkCharacterization}.]
Clearly $P_k$ does not have a connected dominating set satisfying~(\ref{ass:k-short-CDS}).
Hence, (\ref{ass:k-short-CDS}) implies (\ref{ass:Pk-free}).

Conversely, let $H$ be any connected induced subgraph of $G$, and let $X$ be a minimum connected dominating set of $H$.
By Theorem~\ref{thm:sufficiency}, $H[X]$ is $P_{k-2}$-free or $H[X] \cong P_{k-2}$.
If $H[X]$ is $P_{k-2}$-free, the assertion of (\ref{ass:k-short-CDS}) is satisfied.
Otherwise, let $x_1,x_2,\ldots,x_{k-2}$ be a consecutive ordering of the induced path $H[X]$.
In particular, $x_1$ and $x_{k-2}$ are not cut-vertices of $H[X]$.
As $X$ is minimum, there exists a private neighbor $y_i$ of $x_i$, for $i \in \{1,k-2\}$.
It must be that $y_1y_{k-2} \in E(H)$, since otherwise $H[X \cup \{y_1,y_{k-2}\}] \cong P_k$.
Hence, $H[X \cup \{y_1,y_{k-2}\}] \cong C_k$, as desired.
So, (\ref{ass:Pk-free}) implies (\ref{ass:k-short-CDS}).
\end{proof}

\subsection{Proof of Theorem~\ref{thm:algo}}

Before we state our algorithm, we need to introduce some notation and definitions.
For this, let us assume we are given a connected input graph $G$ on $n$ vertices and $m$ edges.
Let $X$ be an arbitrary connected dominating set of $G$.

By ${\it NC}(X)$ we denote the set of vertices in $X$ that are non-cutting in $G[X]$, i.e. for every $x \in {\it NC}(X), G[X\setminus \{x\}] $ is connected.
Let $x$ be a degree-1 vertex of $G[X]$.
We define the \emph{half-path} starting in $x$ to be the maximal path $(x,x_1,x_2,\ldots,x_s)$ in $X$ such that $|N_{G[X]}(x_i)|=2$ for each $i \in \{1,2,\ldots,s-1\}$.
For example, if the neighbor $y \in X$ of $x$ has degree at least 3, the half-path is simply $(x,y)$.
The \emph{length} of the half-path is then $s$.
To each $x \in X$ we assign a weight $w_X(x)$ as follows: 
\begin{enumerate}
  \item if $|N_{G[X]}(x)| \ge 2$, put $w_X(x)=0$, and
  \item if $|N_{G[X]}(x)| = 1$, put $w_X(x)=s$, where $s$ is the length of the half-path starting in $x$.
\end{enumerate}
Finally, the weight $w(X)$ of the set $X$ given by \[w(X) = \sum_{x \in X} (w_X(x))^2.\]
See Fig.~\ref{fig:weights} for an illustration of these definitions.

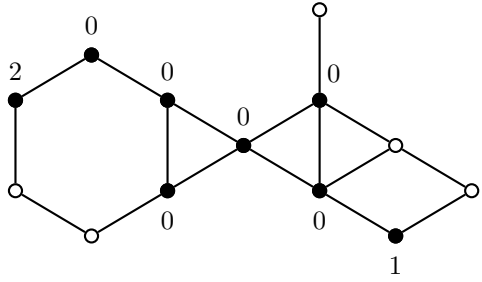
\begin{figure}[ht]
\begin{center}
\psset{yunit=0.6cm}
\begin{pspicture}(0,0)(6,5)

\cnode(0,1){0.1cm}{a}
\cnode(0,3){0.1cm}{b}
\cnode(1,0){0.1cm}{c}
\cnode(1,4){0.1cm}{d}
\cnode(2,1){0.1cm}{e}
\cnode(2,3){0.1cm}{f}
\cnode(3,2){0.1cm}{g}
\cnode(4,1){0.1cm}{h}
\cnode(4,3){0.1cm}{i}
\cnode(4,5){0.1cm}{j}
\cnode(5,0){0.1cm}{k}
\cnode(5,2){0.1cm}{l}
\cnode(6,1){0.1cm}{m}

\nput{90}{b}{$2$}
\nput{90}{d}{$0$}
\nput{90}{f}{$0$}
\nput{-90}{e}{$0$}
\nput{90}{g}{$0$}
\nput{60}{i}{$0$}
\nput{-90}{h}{$0$}
\nput{-90}{k}{$1$}

\pscircle*(0,3){0.1cm}
\pscircle*(1,4){0.1cm}
\pscircle*(2,1){0.1cm}
\pscircle*(2,3){0.1cm}
\pscircle*(3,2){0.1cm}
\pscircle*(4,1){0.1cm}
\pscircle*(4,3){0.1cm}
\pscircle*(5,0){0.1cm}

\ncarc[arcangle=0]{-}{a}{b}
\ncarc[arcangle=0]{-}{a}{c}
\ncarc[arcangle=0]{-}{b}{d}
\ncarc[arcangle=0]{-}{c}{e}
\ncarc[arcangle=0]{-}{d}{f}
\ncarc[arcangle=0]{-}{e}{f}
\ncarc[arcangle=0]{-}{e}{g}
\ncarc[arcangle=0]{-}{f}{g}
\ncarc[arcangle=0]{-}{g}{h}
\ncarc[arcangle=0]{-}{g}{i}
\ncarc[arcangle=0]{-}{h}{i}
\ncarc[arcangle=0]{-}{h}{k}
\ncarc[arcangle=0]{-}{h}{l}
\ncarc[arcangle=0]{-}{i}{j}
\ncarc[arcangle=0]{-}{i}{l}
\ncarc[arcangle=0]{-}{k}{m}
\ncarc[arcangle=0]{-}{l}{m}

\end{pspicture}
\end{center}
\caption{A graph $G$. The black vertices form a connected dominating set $X$ of $G$, with weights $w_X$ as shown. We have $w(X)=5$.}
\label{fig:weights}
\end{figure}

Let $\mathcal X$ be the family of all connected dominating sets of $G$.
We next define a strict partial order $\prec$ on $\mathcal X$ as follows.
For any two sets $X,Y \in \mathcal X$, we put $X \prec Y$ if 
\begin{enumerate}
  \item $|X|>|Y|$, or
  \item $|X|=|Y|$ and $w(X) < w(Y)$.
\end{enumerate}
The \emph{height} of the strict poset $(\mathcal X,\prec)$ is the maximum set of mutually comparable elements of $\mathcal X$.

\begin{lemma}\label{lem:height}
For a connected $n$-vertex graph $G$, the height of $(\mathcal X,\prec)$ is in $\mathcal O(n^3)$.
\end{lemma}
\begin{proof}
If $G[X]$ is not an induced path, every vertex in $X$ of degree at most 2 in $G[X]$ is contained in at most one half-path. 
Hence, $\sum_{x \in X} w_X(x) \le |X|$.
If $G[X]$ is an induced path, every vertex appears in at most two half-paths, implying $\sum_{x \in X} w_X(x) \le 2 |X|$.
Thus \[w(X) = \sum_{x \in X} (w_X(x))^2 \le (\sum_{x \in X} w_X(x))^2 \le 4 |X|^2,\]
and so the weight of a connected dominating set is in $\mathcal O(n^2)$.
Since there are at most $n$ different possible sizes of connected dominating sets of $G$, the height of $(\mathcal X,\prec)$ is in $\mathcal O(n^3)$.
\end{proof}

\begin{proof}[Proof of Theorem~\ref{thm:algo}]
Assume we are given a connected graph $G$ on $n$ vertices and $m$ edges as input.
Our algorithm works as follows, starting with the connected dominating set $Y := V(G)$.
Its output is a connected dominating set $X$ with the properties stated in Theorem~\ref{thm:algo}.
\begin{enumerate}
  \item\label{step:minimize} Compute a minimal connected dominating set $X \subseteq Y$.
  \item\label{step:return-path} If $G[X]$ is an induced path, return $X$ and terminate the algorithm.
  \item\label{step:half-paths} Compute the set ${\it NC}(X)$ and the weight $w_X(x)$ for every $x \in {\it NC}(X)$.
  \item\label{step:ordering} Order the vertices of ${\it NC}(X)$ with non-increasing weight $w_X$, breaking ties arbitrarily. Let that order be $v_1, v_2, \ldots, v_{|{\it NC}(X)|}$.
  \item\label{step:loop} For $i$ from 1 to $|{\it NC}(X)|$ do the following:
  \begin{enumerate}
    \item\label{step:private} Compute a private neighbor $y_i$ of $v_i$ w.r.t.~$X$.
    \item\label{step:innerloop} For $j$ from $i+1$ to $|{\it NC}(X)|$ do the following:
    \begin{enumerate}
       \item Check whether $Y_{ij} := (X \cup \{y_i\}) \setminus \{v_j\}$ is a connected dominating set.
       \item\label{step:substitution} If yes, put $X \leftarrow Y_{ij}$ and go to Step~\ref{step:minimize}.
     \end{enumerate} 
   \end{enumerate}
  \item\label{step:return} Return $X$ and terminate the algorithm.
\end{enumerate}
We remark that the computation of $y_i$ in Step~\ref{step:private} is always possible, since $x_i$ is non-cutting in $G[X]$ and $X$ is a minimal connected dominating set. 
The proof is completed by the following sequence of claims.

\begin{claim}
When the algorithm terminates, the output $X$ is a connected dominating set and $G[X]$ is $P_{k-2}$-free or $G[X] \cong P_{k-2}$.
\end{claim}

%To prove this claim, assume that the algorithm terminates with the return of the set $X$.
Since Step~\ref{step:minimize} is applied before the return is called, $X$ is a minimal connected dominating set.
If the algorithm terminates with Step~\ref{step:return-path}, $G[X]$ is $P_{k-1}$-free by Lemma~\ref{lem:PkCk}.
Hence, either $G[X] \cong P_{k-2}$ or $G[X]$ is $P_{k-2}$-free.

Now assume that the algorithm terminates in Step~\ref{step:return}.
In particular, $G[X]$ is not an induced path.
Suppose for a contradiction that $G[X]$ contains an induced $P_{k-2}$, say on the vertices $x_1,x_2,\ldots,x_{k-2}$.
Like in the proof of Lemma~\ref{lem:substitution}, both $x_1$ and $x_{k-2}$ cannot be cut-vertices of $G[X]$.
Thus, $x_1,x_{k-2} \in {\it NC(X)}$.

After Step~\ref{step:ordering}, the vertices of ${\it NC}(X)$ are ordered $v_1, v_2, \ldots, v_{|{\it NC}(X)|}$ with non-increasing weight.
W.l.o.g.~$x_1 = v_i$, $x_{k-2} = v_j$, and $i < j$.
As $X$ is returned, the set $Y_{ij} := (X \cup \{y_i\}) \setminus \{v_j\}$ is not a connected dominating set, in contradiction to Lemma~\ref{lem:substitution}.
This proves our claim.

\medskip

\begin{claim}\label{claim:poset}
Let $X$ be a minimal connected dominating set considered in some iteration of the algorithm.
Assume that the 'go to' is called in Step~\ref{step:substitution} because $Y_{ij} := (X \cup \{y_i\}) \setminus \{v_j\}$ is a connected dominating set.
Let $X'$ be the minimal connected dominating set computed in the subsequent Step~\ref{step:minimize}.
%If $G[X']$ is not an induced path, 
Then $X \prec X'$.
\end{claim}

Clearly $|X'| \le |X|$.
If $|X'| < |X|$, $X \prec X'$ by definition.
So we may assume that $|X'| = |X|$, and hence $X' = Y_{ij}$.
It remains to show that $w(X) < w(X')$.

Let $z \in X \setminus \{v_i,v_j\}$ be a degree-1 vertex of $G[X]$, and let $(z,x_1,x_2,\ldots,x_s)$ be a half-path starting in $z$.
As $G[X]$ is not a path, $x_s$ is a cut-vertex of $G[X]$.
In particular, $x_s \neq v_j$.
Hence, in $G[Y_{ij}]$, $(z,x_1,x_2,\ldots,x_s)$ is the initial segment of a half-path starting in $z$.
In particular, $w_{X'}(z) \ge w_X(z)$.

If $v_i$ is not a degree-1 vertex of $G[X]$, $w_{X'}(v_i)=w_{X}(v_i)=0$, and $(y_i,v_i)$ is the initial segment of a half-path starting in $y_i$.
Hence, $w_{X'}(y_i) \ge 1$, and thus 
\begin{equation}\label{eqn:weights}
w_{X'}(v_i) = 0 \mbox{ and } w_{X'}(y_i) \ge w_{X}(v_i)+1.
\end{equation}
If the degree of $v_i$ in $G[X]$ is 1, let $(v_i,x_1,x_2,\ldots,x_s)$ be a half-path starting in $v_i$.
Again, $x_s$ is a cut-vertex of $G[X]$, and so $x_s \neq v_j$.
Hence, in $G[X']$, $(y_i,v_i,x_1,x_2,\ldots,x_s)$ is the initial segment of a half-path starting in $y_i$.
Again~(\ref{eqn:weights}) holds.

Summing up, we see that (\ref{eqn:weights}) holds, and 
\begin{equation}\label{eqn:weights2}
w_{X'}(z) \ge w_X(z) \text{ for every vertex }z \in X' \setminus \{y_i,v_i\}.
\end{equation}

We now turn to the vertex $v_j$.
First assume that the degree of $v_j$ in $G[X]$ is at least 2, and thus $w_X(v_j)=0$.
Then, by (\ref{eqn:weights2}), \[w(X') - w(X) \ge w_{X'}(y)^2 - w_{X}(v_j)^2 > 0,\]
and so $w(X') - w(X) > 0$.

Now assume that $v_j$ is a vertex of degree 1 in $G[X]$, and so $w_X(v_j) \ge 1$.
Let $N_{G[X]}(v_j)=\{x\}$.
As $G[X]$ is not a path, $|N_{G[X]}(x)| \ge 2$, and so $w_X(x)=0$.
Thus $w_{X'}(x) = w_{X}(v_j)-1$.
Recall that (\ref{eqn:weights2}) holds, and $w_{X'}(z) \ge w_X(z)$ for every vertex $z \in X' \setminus \{y_i,v_i\}$.
We obtain the following inequality.
\begin{align*}
w(X') - w(X) 	& \ge w_{X'}(y_i)^2 + w_{X'}(x)^2 - w_{X}(v_i)^2 - w_{X}(v_j)^2\\
							& = (w_{X'}(y_i)^2 - w_{X}(v_i)^2) - (w_{X}(v_j)^2 - w_{X'}(x)^2)\\
							& \ge [(w_X(v_i)+1)^2 - w_{X}(v_i)^2] - [w_{X}(v_j)^2 - (w_{X}(v_j)-1)^2]
\end{align*}
But $w_X(v_i) \ge w_X(v_j)$ implies \[(w_X(v_i)+1)^2 - w_{X}(v_i)^2 > w_{X}(v_j)^2 - (w_{X}(v_j)-1)^2,\] and thus $w(X') - w(X) > 0$ holds as in the previous case.

Hence, $X \prec X'$, proving our claim.
\medskip

See Fig.~\ref{fig:weightschange} for an illustration of Step~\ref{step:substitution}.

\begin{figure}[ht]
\begin{center}
\psset{xunit=0.8cm,yunit=0.48cm}
\begin{pspicture}(0,0)(14,5)

\cnode(0,1){0.1cm}{a}
\cnode(0,3){0.1cm}{b}
\cnode(1,0){0.1cm}{c}
\cnode(1,4){0.1cm}{d}
\cnode(2,1){0.1cm}{e}
\cnode(2,3){0.1cm}{f}
\cnode(3,2){0.1cm}{g}
\cnode(4,1){0.1cm}{h}
\cnode(4,3){0.1cm}{i}
\cnode(4,5){0.1cm}{j}
\cnode(5,0){0.1cm}{k}
\cnode(5,2){0.1cm}{l}
\cnode(6,1){0.1cm}{m}

\nput{90}{b}{$2$}
\nput{90}{d}{$0$}
\nput{90}{f}{$0$}
\nput{-90}{e}{$0$}
\nput{90}{g}{$0$}
\nput{60}{i}{$0$}
\nput{-90}{h}{$0$}
\nput{-90}{k}{$1$}

\pscircle*(0,3){0.1cm}
\pscircle*(1,4){0.1cm}
\pscircle*(2,1){0.1cm}
\pscircle*(2,3){0.1cm}
\pscircle*(3,2){0.1cm}
\pscircle*(4,1){0.1cm}
\pscircle*(4,3){0.1cm}
\pscircle*(5,0){0.1cm}

\ncarc[arcangle=0]{-}{a}{b}
\ncarc[arcangle=0]{-}{a}{c}
\ncarc[arcangle=0]{-}{b}{d}
\ncarc[arcangle=0]{-}{c}{e}
\ncarc[arcangle=0]{-}{d}{f}
\ncarc[arcangle=0]{-}{e}{f}
\ncarc[arcangle=0]{-}{e}{g}
\ncarc[arcangle=0]{-}{f}{g}
\ncarc[arcangle=0]{-}{g}{h}
\ncarc[arcangle=0]{-}{g}{i}
\ncarc[arcangle=0]{-}{h}{i}
\ncarc[arcangle=0]{-}{h}{k}
\ncarc[arcangle=0]{-}{h}{l}
\ncarc[arcangle=0]{-}{i}{j}
\ncarc[arcangle=0]{-}{i}{l}
\ncarc[arcangle=0]{-}{k}{m}
\ncarc[arcangle=0]{-}{l}{m}

\cnode(8,1){0.1cm}{a}
\cnode(8,3){0.1cm}{b}
\cnode(9,0){0.1cm}{c}
\cnode(9,4){0.1cm}{d}
\cnode(10,1){0.1cm}{e}
\cnode(10,3){0.1cm}{f}
\cnode(11,2){0.1cm}{g}
\cnode(12,1){0.1cm}{h}
\cnode(12,3){0.1cm}{i}
\cnode(12,5){0.1cm}{j}
\cnode(13,0){0.1cm}{k}
\cnode(13,2){0.1cm}{l}
\cnode(14,1){0.1cm}{m}

\nput{-90}{a}{$4$}
\nput{90}{b}{$0$}
\nput{90}{d}{$0$}
\nput{90}{f}{$0$}
\nput{90}{g}{$0$}
\nput{60}{i}{$0$}
\nput{-90}{h}{$0$}
\nput{-90}{k}{$1$}

\pscircle*(8,1){0.1cm}
\pscircle*(8,3){0.1cm}
\pscircle*(9,4){0.1cm}
\pscircle*(10,3){0.1cm}
\pscircle*(11,2){0.1cm}
\pscircle*(12,1){0.1cm}
\pscircle*(12,3){0.1cm}
\pscircle*(13,0){0.1cm}

\ncarc[arcangle=0]{-}{a}{b}
\ncarc[arcangle=0]{-}{a}{c}
\ncarc[arcangle=0]{-}{b}{d}
\ncarc[arcangle=0]{-}{c}{e}
\ncarc[arcangle=0]{-}{d}{f}
\ncarc[arcangle=0]{-}{e}{f}
\ncarc[arcangle=0]{-}{e}{g}
\ncarc[arcangle=0]{-}{f}{g}
\ncarc[arcangle=0]{-}{g}{h}
\ncarc[arcangle=0]{-}{g}{i}
\ncarc[arcangle=0]{-}{h}{i}
\ncarc[arcangle=0]{-}{h}{k}
\ncarc[arcangle=0]{-}{h}{l}
\ncarc[arcangle=0]{-}{i}{j}
\ncarc[arcangle=0]{-}{i}{l}
\ncarc[arcangle=0]{-}{k}{m}
\ncarc[arcangle=0]{-}{l}{m}

\end{pspicture}
\end{center}
\caption{Before (left) and after (right) an application of Step~\ref{step:substitution}. In the next iteration, the algorithm terminates with the right connected dominating set as output.}
\label{fig:weightschange}
\end{figure}

\begin{claim}
The algorithm terminates in $\mathcal O (n^5(n+m))$ time.
\end{claim}

By Claim~\ref{claim:poset}, each call of the 'go to'-step and the subsequent application of Step~\ref{step:minimize} result in a connected dominating set that is properly larger in the order $\prec$.
By Lemma~\ref{lem:height}, the height of the poset $(\mathcal X,\prec)$, and hence the number of iterations the whole algorithm performs, is in $\mathcal{O}(n^3)$.

It remains to discuss the complexity of the particular steps.
For this, recall that it can be checked in time $\mathcal O(n+m)$ whether a given vertex subset is a connected dominating set.
Consequently, Step~\ref{step:minimize} can be performed in time $\mathcal O(n(n+m))$ by the immediate greedy procedure. 

Step~\ref{step:return-path} and the computation of the weights in Step~\ref{step:half-paths} can both be performed in linear time using the degree sequence of $G[X]$.
The computation of the set ${\it NC}(X)$ in Step~\ref{step:half-paths} can be done straightforwardly in time $\mathcal O(n(n+m))$.

It remains to discuss the complexity of the loop of Step~\ref{step:loop}.
The computation of a private neighbor in Step~\ref{step:private} is clearly done in $\mathcal O (n+m)$ time.
The inner loop of Step~\ref{step:innerloop} consumes $\mathcal O (n)$ checks whether some vertex set is a connected dominating set, requiring $\mathcal O (n+m)$ time each.
Hence, Step~\ref{step:loop} can be done in $\mathcal O (n^2(n+m))$ time.

The overall running time amounts to $\mathcal O (n^5(n+m))$, which completes the proof of both our claim and Theorem~\ref{thm:algo}.
\end{proof}

\subsection{Proof of Theorem~\ref{thm:2-coloring}}

\begin{proof}[Proof of Theorem~\ref{thm:2-coloring}]
Let $H=(V,E)$ be a hypergraph whose incidence graph is $P_7$-free.
A 2-coloring of $H$ we denote by $(A,B)$, where $A,B \subseteq V$ are two non-empty sets with $A \cup B = V$, each of which intersects every hyperedge.

A hypergraph for which any two hyperedges are not comparable (w.r.t.~inclusion) is called a \emph{clutter}.
The following observation was proven by van 't Hof and Paulusma~\cite{HP10}.
In order to be self-contained, we give a quick proof of it.

\begin{claim}\label{claim:clutter}
We may assume that $H$ is a clutter.
\end{claim}
\begin{proof}
Assume there are hyperedges $e,f \in E$ with $e \subseteq f$.
Such a pair of hyperedges we can detect in polynomial time.

Every 2-coloring of $H$ is a 2-coloring of the hypergraph $H' = (V, E \setminus \{f\})$ in particular.
If $(A,B)$ is a 2-coloring of $H'$, it holds that $e \cap A \neq \emptyset$ and $e \cap B \neq \emptyset$.
Thus, $f \cap A \neq \emptyset$ and $f \cap B \neq \emptyset$, and so $(A,B)$ is a 2-coloring of $H$.

So we may delete, for every such pair $e,f \in E$ with $e \subseteq f$ the hyperedge $f$ from $H$.
It is clear that the resulting hypergraph is a clutter, and its incidence graph is still $P_7$-free.
This proves Claim \ref{claim:clutter}.
\end{proof}

Although immediate, Claim~\ref{claim:clutter} considerably simplifies the argumentation of the following proof.
We now assume that $H$ is a clutter.
Moreover, we may assume that $H$ is connected, that is, its incidence graph is connected.
In the following, we prove a sequence of claims that discuss all relevant cases for the 2-coloring problem.
We state the polynomial algorithm along the way.

Let $G$ be the incidence graph of $H$.
Since we are searching for a 2-coloring, we may assume that $|N_G(f)|\ge 2$ for every $f \in E$.
By Theorem~\ref{thm:PkCharacterization}, there is a connected dominating set $X$ of $G$ such that $G[X]$ is $P_5$-free or $G[X] \cong C_7$.
However, the latter case contradicts the fact that $G$ is bipartite. 
So, $G[X]$ is a connected $P_5$-free graph.

Using Theorem~\ref{thm:PkCharacterization} again, we see that $G[X]$ has a dominating $P_3$-free graph.
That is, there is a pair of adjacent vertices, say $v \in V$ and $e \in E$, that together dominate $G[X]$.
In particular, $e$ intersects every other hyperedge.
It is clear that we can compute such hyperedge in polynomial time.

%Since $H$ is a connected clutter, we may assume that every hyperedge has at least two elements.
%Moreover, since the \textsc{Hypergraph 2-Colorability} is solvable in linear time for graphs (it is equivalent to testing whether $H$ is a bipartite graph), we can assume that there is a hyperedge of size at least 3, say $e$.
%
%
%$e$ intersects every hyperedge and is therefore a transversal of $H$.

\begin{claim}\label{clm minimality}
If there is a proper subset $X \subset e$ that dominates $E$, $(X, V \setminus X)$ is a 2-coloring of $H$.
\end{claim}

Let $f \in E$ be arbitrary.
By assumption, $f \cap X \neq \emptyset$.
Since $H$ is a clutter, $f \not\subseteq e$, and thus $f \not \subseteq X$.
Hence, $f \setminus X \neq \emptyset$, proving Claim~\ref{clm minimality}.\\

Indeed, it can be checked in polynomial time whether there is a proper subset $X \subset e$ that dominates $E$.
(If so, $X$ is found in polynomial time, too.)
In view of Claim \ref{clm minimality}, we may assume that no proper subset of $e$ dominates $E$.

We now make a distinction of the cases $|e|=2$ and $|e| \ge 3$.
Let us first assume that $|e| = 2$, say $e = \{x,y\}$.
Since $H$ is a clutter, every hyperedge $f$ of $H$ contains either $x$ or $y$.
Let $X,Y \subseteq E \setminus \{e\}$ such that every $f \in X$ contains $x$, every $g \in Y$ contains $y$, and $X \cup Y = E \setminus \{e\}$.

If $|X|=0$, every hyperedge contains $y$ and, as $H$ is a clutter, some other vertex.
Thus a 2-coloring of $H$ is given by $(\{y\},V \setminus \{y\})$.
By symmetry, we may now assume that $|X|,|Y| \ge 1$.
Observe that, if $|X|=1$, say $X=\{f\}$, $H$ is 2-colorable if and only if there is some vertex $v \in f$ such that $(\{v,y\},V\setminus\{v,y\})$ is a 2-coloring of $H$. 
Indeed, if for every vertex $v \in f$, $(\{v,y\}, V \setminus \{v,y\})$ is not a 2-coloring of $H$, there exists a hyperedge $e_v =\{v,y\}$ for each such vertex $v$. 
%Note that a such vertex $v$ must exist because $H$ is a clutter. 
Let now $v$ be an arbitrary vertex in $f$. 
Then $x$ and $v$ must have the same color, and so there is a vertex $v' \in f$ with the second color. 
Then, the hyperedge $e_{v'} = \{v',y\}$ is monochromatic, a contradiction.
This condition can clearly be checked in polynomial time.

Now let $|X|,|Y|\ge 2$.
We next show that $H$ admits a 2-coloring.
To see this, pick any $f \in X$ and $g \in Y$.
Since $H$ is a clutter, $f \setminus e, g \setminus e \neq \emptyset$.
Pick any $u \in f \setminus e$ and $v \in g \setminus e$.
If $fv,gu \not\in E(G)$, $G[\{u,f,x,e,y,g,v\}] \cong P_7$, a contradiction.
As $u$ and $v$ were arbitrary, it must be that $f \setminus e \subseteq g \setminus e$ or $g \setminus e \subseteq f \setminus e$.

Now let $f,f' \in X$ and $g \in Y$ be three mutually distinct hyperedges.
As shown above, the sets $f \setminus e, f' \setminus e$ are comparable to $g \setminus e$.
Since $H$ is a clutter, $f \setminus e$ is not comparable to $f' \setminus e$.
Hence, either $f \setminus e, f' \setminus e \subseteq g$, or $g \setminus e \subseteq f, f'$.

In the first case, $f \setminus e \subseteq g$ for any $f \in X,g\in Y$.
Thus, $(\bigcup_{f \in X} f) \setminus e \subseteq \bigcap_{g \in Y} g$.
Since $H$ is a clutter, every $g \in Y$ has a neigbor outside the set $\{y\} \cup \bigcap_{g \in Y} g$.
Hence, \[(\{y\} \cup \bigcap_{g \in Y} g, V \setminus (\{y\} \cup \bigcap_{g \in Y} g))\] is a 2-coloring of $H$.
The second case, $g \setminus e \subseteq f, f'$, is dealt with in a similar fashion.

So we may assume $|e|\ge 3$.
Since no proper subset of $e$ dominates $E$ in $G$, the following holds:
for every $x \in e$ there is a hyperedge $f_x$ such that $f_x \cap e = \{x\}$.

\begin{claim}\label{clm equal outside}
For all $x,y \in e$, $f_x \setminus e = f_y \setminus e$.
\end{claim}

Let $x,y \in e$.
The case that $x = y$ is trivial.
So we may assume that $x \neq y$.

Suppose that there is a vertex $z \in f_x \setminus (e \cup f_y)$.
If there is a vertex $z' \in f_y \setminus (e \cup f_x)$, $G[\{z,f_x,x,e,y,f_y,z'\}] \cong P_7$, a contradiction.
Thus, $f_y \setminus (e \cup f_x) = \emptyset$, and so $f_y \setminus e \subseteq f_x \setminus e$.

Since $H$ is a clutter, there is a vertex $u \in f_y \setminus e$.
As $f_y \setminus e \subseteq f_x \setminus e$, $u \in (f_x \cap f_y) \setminus e$.
Since $|e| \ge 3$, there is a vertex $v \in e \setminus \{x,y\}$.
But then $G[\{z,f_x,u,f_y,y,e,v\}] \cong P_7$, a contradiction.

So, $f_x \setminus (e \cup f_y) = \emptyset$ and, for symmetry, $f_y \setminus (e \cup f_x) = \emptyset$.
This proves Claim \ref{clm equal outside}.
For an illustration, see Fig.~\ref{fig:outside}.

\begin{figure}[ht]
\begin{center}
\psset{xunit=1.2cm}
\begin{pspicture}(1,0)(5,2.5)

\cnode(1,0){0.1cm}{z}
\cnode(1,2){0.1cm}{fx}
\cnode(2,0){0.1cm}{u}
\cnode(3,2){0.1cm}{fy}
\cnode(3,0){0.1cm}{x}
\cnode(4,0){0.1cm}{y}
\cnode(5,2){0.1cm}{e}
\cnode(5,0){0.1cm}{v}

\nput{90}{fy}{$f_y$}
\nput{90}{e}{$e$}
\nput{90}{fx}{$f_x$}

\nput{-90}{u}{$u$}
\nput{-90}{v}{$v$}
\nput{-90}{y}{$y$}
\nput{-90}{x}{$x$}
\nput{-90}{z}{$z$}

\ncarc[arcangle=0]{-}{fy}{u}
\ncarc[arcangle=0]{-}{fy}{y}
\ncarc[arcangle=0]{-}{fx}{u}
\ncarc[arcangle=0]{-}{fx}{x}
\ncarc[arcangle=0]{-}{fx}{z}
\ncarc[arcangle=0]{-}{e}{x}
\ncarc[arcangle=0]{-}{e}{y}
\ncarc[arcangle=0]{-}{e}{v}

\end{pspicture}
\end{center}
\caption{The situation in the proof of Claim~\ref{clm equal outside}.}
\label{fig:outside}
\end{figure}
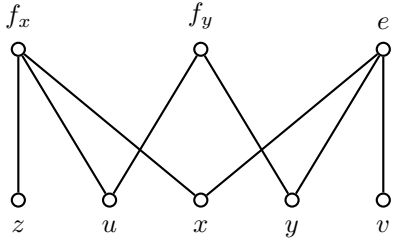

\begin{claim}\label{claim:small-hyperedges}
If $|f_x \setminus e| = 1$ for some $x \in e$, $H$ does not admit a 2-coloring.
\end{claim}

Assume that $|f_x \setminus e| = 1$ for some $x \in e$.
By Claim \ref{clm equal outside}, there is a vertex $v \in V$ such that $f_y \setminus e = \{v\}$ for all $y \in e$.

Suppose that $(A,B)$ is a 2-coloring of $H$.
We may assume that $v \in A$.
Since for every $z \in e$, $f_z \cap B \neq \emptyset$, $e \subseteq B$ holds, a contradiction.
So Claim \ref{claim:small-hyperedges} holds.\\

It can be checked in polynomial time whether $|f_x \setminus e| = 1$ for some $x \in e$.
In view of Claim \ref{clm equal outside} and Claim \ref{claim:small-hyperedges}, we may now assume that $|f_x \setminus e| \ge 2$ for all $x \in e$.

\begin{claim}\label{clm explicite 2-coloring}
Let $x,y \in e$ be two arbitrary, distinct vertices and let $z \in f_x \setminus e$.
A 2-coloring of $H$ is given by $(\{x,y,z\},V \setminus \{x,y,z\})$.
\end{claim}

Let $x,y,z$ be chosen according to the claim.
Suppose that $(\{x,y,z\},V \setminus \{x,y,z\})$ is not a 2-coloring of $H$.
Thus there is an hyperedge $f$ with $f \subseteq \{x,y,z\}$ or $f \cap \{x,y,z\} = \emptyset$.

Let us first assume $f \subseteq \{x,y,z\}$.
In particular, $|f \setminus e| \le 1$.
Since $|f_{x'} \setminus e| \ge 2$ for all $x' \in e$, we know that $|f \cap e| \neq 1$.
As $N_G(e)$ dominates the set $E$, $|f \cap e| \ge 2$ and so $x,y \in f$.
Since $H$ is a clutter, $f \not\subseteq e$, and so $f = \{x,y,z\}$.

By assumption, $|f_x \setminus e| \ge 2$, and so there is a vertex $z' \in f_x \setminus (e \cup \{z\})$.
Moreover, since $|e| \ge 3$, there is a vertex $x' \in E \setminus \{x,y\}$.
But then $G[\{z',f_x,z,f,y,e,x'\}] \cong P_7$, a contradiction.

So we may assume $f \cap \{x,y,z\} = \emptyset$.
As $N_G(e)$ dominates $E$, $e \cap f \neq \emptyset$.
Let $x' \in e \cap f$.
As $H$ is not a clutter, there is some $z' \in f \setminus e$.
This situation is illustrated in Fig.~\ref{fig:{x,y,z}dominates}.

\begin{figure}[ht]
\begin{center}
\psset{xunit=1.2cm}
\begin{pspicture}(1,0)(5,2.5)

\cnode(1,0){0.1cm}{z'}
\cnode(1,2){0.1cm}{f}
\cnode(2,0){0.1cm}{x'}
\cnode(3,2){0.1cm}{e}
\cnode(3,0){0.1cm}{y}
\cnode(4,0){0.1cm}{x}
\cnode(5,2){0.1cm}{fx}
\cnode(5,0){0.1cm}{z}

\nput{90}{f}{$f$}
\nput{90}{e}{$e$}
\nput{90}{fx}{$f_x$}

\nput{-90}{z'}{$z'$}
\nput{-90}{x'}{$x'$}
\nput{-90}{y}{$y$}
\nput{-90}{x}{$x$}
\nput{-90}{z}{$z$}

\ncarc[arcangle=0]{-}{z'}{f}
\ncarc[arcangle=0]{-}{f}{x'}
\ncarc[arcangle=0]{-}{x'}{e}
\ncarc[arcangle=0]{-}{e}{y}
\ncarc[arcangle=0]{-}{e}{x}
\ncarc[arcangle=0]{-}{x}{fx}
\ncarc[arcangle=0]{-}{fx}{z}
\ncarc[arcangle=0,linestyle=dashed]{-}{fx}{z'}

\end{pspicture}
\end{center}
\caption{The situation in the proof of Claim~\ref{clm explicite 2-coloring}. The dashed edge is optional.}
\label{fig:{x,y,z}dominates}
\end{figure}
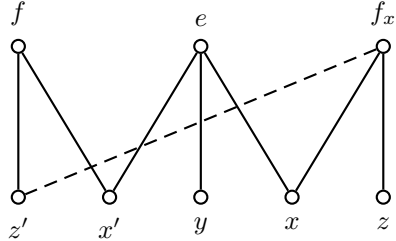

If $f_xz' \in E(G)$, $G[\{z,f_x,z',f,x',e,y\}] \cong P_7$, a contradiction.
Otherwise, $G[\{z,f_x,x,e,x',f,z'\}] \cong P_7$, another contradiction.
This proves Claim~\ref{clm explicite 2-coloring}.\\

Clearly, a 2-coloring as provided by Claim \ref{clm explicite 2-coloring} can be constructed efficiently.
This completes the proof.
\end{proof}

\section{Conclusion}

In this paper we gave a description of the structure of connected dominating sets in $P_k$-free graphs.
We have shown that any connected $P_k$-free graph admits a connected dominating set whose induced subgraph is $P_{k-2}$-free or isomorphic to $P_{k-2}$.
In fact, any minimum connected dominating set has this property.
Loosely speaking, this means that the restricted structure of connected $P_k$-free graphs results in an even more restricted structure of the induced subgraph of their minimum connected dominating sets.

Although we think that our results are of their own interest, our hope is that they might be useful in other contexts, too.
One example we gave is the polynomial time solvability of \textsc{Hypergraph 2-Colorability} for hypergraphs with $P_7$-free incidence graph. 
It seems possible that, with more work, one could push this result to hypergraphs with $P_8$-free incidence graph.
However, more interesting would be to know whether there is any $k$ for which \textsc{Hypergraph 2-Colorability} for hypergraphs with $P_k$-free incidence graph is \emph{not} solvable in polynomial time.
So far, we do not have an opinion or an intelligent guess on this question.
%at nor are we aware of an NP-hard problem that we could reduce to this restricted version of \textsc{Hypergraph 2-Colorability}.

Other possible future applications of our results include the coloring of $P_k$-free graphs.
As mentioned earlier, Hoang~\emph{et al.}~\cite{HKLSS65} showed that $k$-\textsc{Colorability} is efficiently solvable on $P_5$-free graphs, using the fact that a connected $P_5$-free graph has a dominating clique or a dominating induced $P_3$.
To our knowledge, an open problem, conjectured by Huang~\cite{Hua13}, in this context is whether 4-colorability can be decided in polynomial time for $P_6$-free graphs.
From Theorem~\ref{thm:algo} it follows that, given a $P_6$-free graph, we can efficiently compute a connected dominating set that induces a $P_4$-free graph (that is a cograph) or a $P_4$.
Of course cographs are less trivial than cliques, especially when it comes to coloring -- but that does not rule out an approach similar to that of Ho\'ang~\emph{et al.}~\cite{HKLSS65}. 
The fact that each vertex of the graph has some neighbor in this cograph leaves a 3-coloring problem for the rest of the graph, once the coloring of the cograph is fixed.
Here, one might use the fact that 3-coloring is polynomial time solvable for $P_6$-free graphs, shown by Randerath and Schiermeyer~\cite{RS04}, even in the pre-coloring extension version, proven by Broersma~\emph{et al.}~\cite{BFGP13}.

%\nocite{HP10b}

\end{document}